\documentclass[11pt]{amsart}
\usepackage{amsfonts}
\usepackage{amssymb}
\usepackage{graphicx}

\def\a{\alpha}       \def\b{\beta}        
           
\def\D{{\mathbb D}}  
\def\C{{\mathbb C}}  
\def\R{{\mathbb R}}

\def\({\left(}       \def\){\right)}

\newtheorem{lem}{\sc Lemma}
\newtheorem{thm}{\sc Theorem}

\newtheorem{ex}{\sc Example}

\newtheorem{rk}{\sc Remark}
\newenvironment{pf}{\noindent{\textit{Proof. }}}{$\Box$ }

\newenvironment{dedication}
        {\begin{center}\begin{em}}
        {\par\end{em}\end{center}}
\begin{document}
\title[A harmonic maps approach to fluid flows]{A harmonic maps approach to fluid flows}

\author[O. Constantin]{Olivia Constantin}
\address{Faculty of Mathematics, University of Vienna, Oskar-Morgenstern-Platz~1, 1090 Vienna, Austria.} \email{olivia.constantin@univie.ac.at}

\author[M. J. Mart\'{\i}n]{Mar\'{\i}a J. Mart\'{\i}n}
\address{Department of Physics and Mathematics, University of Eastern Finland, P.O. Box 111, FI-80101 Joensuu, Finland.} \email{maria.martin@uef.fi}

\thanks{The first author was supported by FWF project P 24986-N25. The second author was supported in part by Academy of Finland grant 268009 and by Spanish MINECO Research Project MTM2015-65792-P}
\maketitle
\begin{dedication}
Dedicated to Professor Peter Duren on the occasion of his 80th birthday
\end{dedication}
\begin{abstract}
We obtain a complete solution to the problem of classifying all two-dimensional ideal fluid flows with harmonic Lagrangian labelling maps; thus, we explicitly provide all solutions, with the specified structural property, to the incompressible two-dimensional Euler equations (in Lagrangian variables).
\end{abstract}

\section{Introduction}

Within the realm of fluid dynamics, the most complete description of a flow is attained within the Lagrangian framework.  Qualitative studies of fluids
are concerned with perturbations of explicitly known flows. Therefore non-trivial explicit
solutions that capture important aspects of the physical reality are of great importance \cite{Bennet}.
Several publications in the applied mathematics, engineering and physics research literature (see
 \cite{Go}, \cite{H}, \cite{HH}, \cite{IK}, and \cite{L} and references therein) exploited a remarkable feature
shared by some celebrated explicit solutions to the two-dimensional incompressible Euler equations, in
Lagrangian variables (such as Kirchhoff's elliptical vortex \cite{Kirchhoff} found in 1876, Gerstner's flow \cite{Gerstner} found in 1809
and re-discovered in 1863 by Rankine \cite{Rankine}, and the Ptolemaic vortices found in 1984 by Abrashkin and Yakubovich
\cite{A-Y}), namely that in all of them the labelling map is harmonic at all times.
\par
Recently, in \cite{A-C}, the authors
proposed a complex analysis approach aimed at classifying all such flows. While new explicit solutions were
obtained, the exhaustion of all possibilities was reduced to  an explicit nonlinear ordinary differential
system in $\mathbb{C}^4$. Solving this system in full generality proved elusive so far.
\par
We propose a different approach that provides a complete  solution to the original problem of finding  all flows with harmonic labelling maps.
Our approach is  based on ideas from
the theory of harmonic mappings,
more precisely, on the fact that
it is possible to characterize
the relationship between planar harmonic maps having the same
Jacobian --- a property of the labelling maps that is a consequence of the equation of mass conservation, expressed in
Lagrangian variables. Apart from achieving the full picture, our considerations provide an illustration of the deep links between the fields of complex analysis and fluid mechanics.

\section{The governing equations}

The Eulerian description of the two-dimensional motion of an ideal homogeneous fluid is obtained by imposing
the law of mass conservation
\begin{equation}\label{mc}
u_x+v_y=0,
\end{equation}
and Euler's equation of motion
\begin{equation}\label{eu}
\left\{\begin{array}{ccc}
u_t + uu_x+vu_y &=& -\,P_x,\\[0.2cm]
v_t + uv_x+vv_y &=& -\,P_y,
\end{array}\right.
\end{equation}
where $\big(u(t,x,y),\,v(t,x,y)\big)$ is the velocity field in the time and space variables $(t,x,y)$ and
the scalar function $P(t,x,y)$ represents the pressure. Since the reference density in hydrodynamics is $1\, g/cm^3$, we normalize the constant fluid density to 1.

The most complete flow representation is provided in (material) La\-gran\-gian coordinates,
in which one describes the motion of all fluid particles. For a given velocity field
$\big(u(t,x,y),\,v(t,x,y)\big)$, the motion of the individual
particles $\big(x(t),\,y(t)\big)$ is obtained by integrating the system of ordinary differential equations
$$\left\{\begin{array}{ccc}
x'(t) &=& u(t,x,y),\\[0.15cm]
y'(t) &=& v(t,x,y),
\end{array}\right.$$
whereas the knowledge of the particle path $t \mapsto \big(x(t),\,y(t)\big)$ provides by differentiation
with respect to $t$  the velocity field at time $t$ and at the location $\big(x(t),\,y(t)\big)$.
\par
Starting with a simply connected domain $\Omega_0$, representing the labelling domain, each
label $(a,b) \in \Omega_0$ identifies by means of the injective map
\begin{equation}\label{lab}
(a,b) \mapsto \big(x(t;a,b),\,y(t;a,b)\big)
\end{equation}
the evolution in time of a specific particle, the
fluid domain at time $t$, $\Omega(t)$, being the image of $\Omega_0$ under the map (\ref{lab}).
To write the governing equations in Lagrangian coordinates, we use the following relations:
\[
\left\{\begin{array}{ccc}
\displaystyle\frac{\partial}{\partial a} &=& x_a\,\displaystyle\frac{\partial}{\partial x}
\,+\,  y_a\,\displaystyle\frac{\partial }{\partial y},\\[0.35cm]
\displaystyle\frac{\partial}{\partial b} &=& x_b\,\displaystyle\frac{\partial }{\partial x}
\,+\,  y_b\,\displaystyle\frac{\partial}{\partial y},
\end{array}\right.
\]
and
\begin{equation}\label{sd2}
\left\{\begin{array}{ccc}
\displaystyle\frac{\partial}{\partial x} &=& \displaystyle\frac{1}{J}\,\Big(y_b\,\displaystyle\frac{\partial }{\partial a}
\,-\,  y_a\,\displaystyle\frac{\partial}{\partial b}\Big),\\[0.35cm]
\displaystyle\frac{\partial}{\partial y} &=& \displaystyle\frac{1}{J}\,\Big(x_a\,\displaystyle\frac{\partial}{\partial b}
\,-\,  x_b\,\displaystyle\frac{\partial}{\partial a}\Big),
\end{array}\right.
\end{equation}
where $J$ is the Jacobian of the transformation given by
\begin{equation}\label{j}
J = \Big|\frac{\partial(x,y)}{\partial (a,b)}\Big|=x_a\,y_b-y_a\,x_b\,\neq 0.
\end{equation}
The local injectivity of the transformation between the Eulerian and Lagrangian coordinates is expressed by
(\ref{j}). In view of (\ref{sd2}) and
\begin{equation}\label{el}
\left\{\begin{array}{ccc}
u(t,x,y) &=& \displaystyle\frac{\partial}{\partial t}\,x(t; a,b),\\[0.35cm]
v(t,x,y) &=& \displaystyle\frac{\partial}{\partial t}\,y(t; a,b),
\end{array}\right.
\end{equation}
the equation of mass conservation (\ref{mc}) takes the form
\begin{eqnarray*}
0 &=& u_x+v_y= \displaystyle\frac{y_b\,x_{at}-y_a\,x_{bt}+x_a\,y_{bt}-x_b\,y_{at}}{J}=\displaystyle\frac{J_t}{J}
\end{eqnarray*}
in Lagrangian coordinates, that is,
\begin{equation}\label{jt}
J_t=0\,.
\end{equation}
On the other hand, from (\ref{el}) we get
\begin{eqnarray*}
x_{tt} &=& u_t + u_xx_t +u_yy_t =u_t + uu_x+vu_y,\\
y_{tt} &=& v_t + uv_x+vv_y\,,
\end{eqnarray*}
so that the Euler equation (\ref{eu}), in Lagrangian variables, becomes
$$\left\{\begin{array}{ccc}
x_{tt} &=&-\,\displaystyle\frac{y_b\,P_a-y_a\,P_b}{J},\\[0.35cm]
y_{tt} &=&-\,\displaystyle\frac{x_a\,P_b-x_b\,P_a}{J}.
\end{array}\right.$$
Since the Jacobian of the above system does not vanish, due to (\ref{j}), we can solve for the gradient
of $P$ in the label space, obtaining
$$\left\{\begin{array}{ccc}
P_a &=&-\,x_a\,x_{tt}\,-\,y_a\,y_{tt},\\[0.15cm]
P_b &=&-\,x_b\,x_{tt}\,-\,y_b\,y_{tt},
\end{array}\right.$$
in $\Omega_0$. The domain  $\Omega_0$ being simply connected, the above system is equivalent to
the compatibility condition $P_{ab}=P_{ba}$, that is,
\[
x_a\,x_{btt} \,+\,y_a\,y_{btt}\,=\,x_b\,x_{att}\,+\,y_b\,y_{att},
\]
or, equivalently,
\begin{equation}\label{eul}
\Big( x_a\,x_{bt}\,+\,y_a\,y_{bt} \,-\,x_b\,x_{at}\,-\,y_b\,y_{at}\Big)_t=0.
\end{equation}
These considerations show that, in Lagrangian coordinates, the governing equations
are equivalent to (\ref{jt}) and (\ref{eul}), under the constraint that, at any instant
$t$, the map (\ref{lab}) is a global diffeomorphism from the label domain
$\Omega_0$ to the fluid domain $\Omega(t)$.

\section{Harmonic labellings}
As already mentioned in Introduction, the common structural property of the known explicit solutions to the
governing equations (Kirchhoff's vortex \cite{Kirchhoff}, Gerstner's wave \cite{Gerstner}, the Ptolemaic solutions \cite{A-Y})
is that the labelling map (\ref{lab}) is harmonic at every
fixed time $t$. Our aim is to explicitly find  all solutions having this property.
Most ideas in this section are inspired by \cite{A-C} and are included for the sake of completeness.
Since from now on our methods will rely
exclusively on complex analysis, it is convenient to adapt the notation accordingly.
Therefore we look for solutions to (\ref{jt}) and (\ref{eul}) having the form
\begin{equation}\label{an}
x(t;\,a,b)\,+\,i\,y(t;\,a,b)=F(t,z) \,+\,\overline{G(t,z)},\qquad z=a+ib,
\end{equation}
with $z \mapsto F(t,z)$ and $z \mapsto G(t,z)$ analytic in the simply connected
domain $\Omega_0 \subset \C$, at every time $t$. Recall that
$\displaystyle\frac{\partial F}{\partial \overline{z}}=0$ characterizes analyticity and that
\begin{equation}\label{dd}
\frac{\partial}{\partial a}=\frac{\partial}{\partial z}\,+\,\frac{\partial}{\partial \overline{z}},\qquad
\frac{\partial}{\partial b}=i\,\Big(\frac{\partial}{\partial z}\,-\,\frac{\partial}{\partial \overline{z}}\Big),
\qquad \displaystyle\frac{\partial \overline{f}}{\partial \overline{z}}=
\overline{\displaystyle\frac{\partial f}{\partial z}}.
\end{equation}
\par
The Jacobian of the harmonic map (\ref{an}) is $J=|F'|^2-|G'|^2$. Therefore the equation of mass
conservation (\ref{jt}) can be re-written as $\big( F'\,\overline{F'} \,-\,\overline{G'}\,G'\big)_t=0$,
that is,
\begin{equation}\label{mclh}
Re\,\big( F'_t\,\overline{F'} \,-\,G'\, \overline{G'_t}\big)=0\,.
\end{equation}

Relation (\ref{dd}) together with (\ref{an}) yield
$$\begin{array}{ccc}
x_a=\displaystyle\frac{F'+G'+ \overline{F'} + \overline{G'}}{2}\, , &\quad&
x_b=i\,\displaystyle\frac{F'+G'- \overline{F'} - \overline{G'}}{2}\,,\\[0.3cm]
y_a=\displaystyle\frac{F'-G'- \overline{F'} + \overline{G'}}{2i}\, , &\quad&
y_b=\displaystyle\frac{F'-G'+ \overline{F'} - \overline{G'}}{2}\,.
\end{array}$$
A lengthy but straightforward calculation using the above relations shows that (\ref{eul})
is equivalent to
\begin{equation}\label{euh}
Im\,\big( F'_t\,\overline{F'} \,-\,G'\, \overline{G'_t}\big)_t=0\,.
\end{equation}
In view of (\ref{mclh}) and (\ref{euh}), the governing equations reduce to the single equation
\begin{equation}\label{ge}
F'_t\,\overline{F'} \,-\,G'\, \overline{G'_t}=i\,\nu(z, \,\bar{z})\,
\end{equation}
for some real-valued function $\nu$.
\begin{rk}
Using the notation $F_0:=F(0, \cdot ),\, G_0:=G(0,\cdot)$, the equation of mass conservation expressed as $J_t=0$ means that the Jacobian of the labelling map is constant in time,
and hence it is given by $J=|F'_0|^2-|G'_0|^2$.  Since $\Omega_0$ is simply connected and $J$ does not vanish in
$\Omega_0$, we deduce that we either have $J>0$ (i.e. $F_0+\overline{G_0}$ is sense preserving) or  $J<0$ (i.e. $F_0+\overline{G_0}$ is sense reversing) throughout
$\Omega_0$. From now on we shall assume without loss of generality that  $F_0+\overline{G_0}$ is sense preserving. Indeed, if  $F_0+\overline{G_0}$ is sense reversing, we can replace the label domain by $\overline{\Omega_0}$, and, in this case, the labelling map at time $t=0$ becomes $F_0(\bar z)+\overline{G_0(\bar z)}$, which is
sense preserving.
\par\smallskip
An equivalent condition to the fact that $F_0+\overline{G_0}$ is sense preserving in the simply connected domain $\Omega_0$ is that $F_0$ is locally univalent (thus $F_0'\neq 0$ in $\Omega_0$) and the (second complex) dilatation $\omega$ of $F_0+\overline{G_0}$, defined by $\omega=G_0'/F_0'$, is an analytic map into the open unit disk.
\end{rk}
\par
Since the Jacobian of the harmonic maps $F(t,z)+\overline{G(t,z)}$, defined for $t\geq 0$ and $z\in \Omega_0$, that satisfy (\ref{ge}) is independent of time, in order to find all solutions it is natural to start by looking for a characterization of the relationship between two harmonic maps with equal Jacobians.

\section{Harmonic functions with equal Jacobian}

We now find the relationship between two harmonic maps with equal Jacobians. Our approach is inspired by the proof of \cite[Thm. 3]{Ch-D-O} (see also the considerations made in \cite{HM-Schwarzian}).

\par\smallskip
We start by proving the following lemma that might have some independent interest and will be used later on.

\begin{lem}\label{lem-log}
Let $\varphi$ and $\psi$ be two analytic functions in a simply connected domain $\Omega$, with $\varphi\not\equiv 0$.  Then
\begin{equation}\label{eq-lemlog}
|\varphi|^2=r|\psi|^2 +s
\end{equation}
in $\Omega$ for some real numbers $r$ and $s$ different from zero if and only if there exist two complex constants $c_1$ and $c_2$ with $|c_1|^2= r|c_2|^2 + s$ such that $\varphi\equiv c_1$ and $\psi\equiv c_2$.
\end{lem}

\begin{pf}
Note that for (\ref{eq-lemlog}) to hold, it is necessary that $r|\psi(z)|^2+s \geq 0$ for all $z\in\Omega$. If $r|\psi|^2+s \equiv 0 $, then a direct application of the open mapping theorem for analytic functions gives us the desired result.
\par
Assume now that $r|\psi|^2+s$ is not identically zero. Then, there exists a disk $D= D(z_0, R)$ centered at some $z_0\in\Omega$ and radius $R>0$ such that $D\subset\Omega$ and $r|\psi(z)|^2+s>0$ for all $z\in D$. The forthcoming analysis will be done in this disk.
\par
We take logarithms in (\ref{eq-lemlog}) to get
\begin{equation}\label{eq-lemlog2}
\log |\varphi|^2=\log\left(r|\psi|^2+s\right)\,.
\end{equation}
Now, the function of the left-hand side in (\ref{eq-lemlog2}) is harmonic so that the one on the right-hand side must be harmonic as well. The Laplacian of this latter function equals
\[
\Delta \log\left(r|\psi|^2+s\right)=4rs \cdot \frac{|\psi'|^2}{(r|\psi|^2+s)^2}\,,
\]
so that $\psi'$ is identically zero in $D$ and hence there exists a non-zero constant $c_2 \in \C$ such that $\psi\equiv c_2$. Bearing in mind (\ref{eq-lemlog}), we get that $\varphi$ equals a constant $c_1$ too, with $|c_1|^2= r|c_2|^2 + s$. A direct application of the identity principle for analytic functions completes the proof.
\end{pf}
\par\smallskip
\subsection{The case of linear dependence}
We first treat the easier case when the harmonic mapping $F+\overline G$ is such that $F'$ and $G'$ are linearly dependent.
\begin{thm}\label{lin-dep}
Let $F_1+\overline{G_1}$ be a (sense preserving) harmonic map
in a simply connected domain $\Omega \subset {\mathbb C}$. Assume further
that $G_1'=\lambda F_1'$ where $\lambda \in {\mathbb C}$. If $F_2+\overline{G_2}$ is a harmonic
map in $\Omega$ whose Jacobian equals that of $F_1+\overline{G_1}$, that is,
\begin{equation}\label{ej}
|F_1'|^2- |G_1'|^2=|F_2'|^2- |G_2'|^2\,,
\end{equation}
then there exist constants $\a,\,\b \in {\mathbb C}$ such that $F_2'=\a F_1'$ and $G_2'=\b F_1'$, with
$|\a|^2-|\b|^2=1-|\lambda|^2>0$. Moreover, if $F_1+\overline{G_1}$ is univalent in $\Omega$,
then $F_2+\overline{G_2}$ is univalent in $\Omega$.
\end{thm}

\begin{pf} Let us first notice that $F_1'$ has no zeros in $\Omega$ since the Jacobian $J_1$ of $F_1+\overline{G_1}$ is strictly positive. Dividing by $|F_1'|^2$ in (\ref{ej}), we get
\[
1-|\lambda|^2=\Big| \frac{F_2'}{F_1'}\Big|^2- \Big| \frac{G_2'}{F_1'}\Big|^2\,.
\]
We can now apply Lemma \ref{lem-log} to deduce that there exist $\a,\,\b \in {\mathbb C}$ such that $F_2'=\a F_1'$ and $G_2'=\b F_1'$, with $1-|\lambda|^2=|\a|^2-|\b|^2$.  In particular, $|\a|^2-|\b|^2 > 0$ since $J_1 > 0$.

It remains to show that univalence is preserved. If $F_1+\overline{G_1}$ is univalent, then, due to the special form of $G_1$, we have that $F_1$ is also univalent. Assume
there exist $z, \, w \in \Omega$ with $F_2(z)+\overline{G_2(z)}=F_2(w)+\overline{G_2(w)}$.
Since $F_2+\overline{G_2}=\a F_1 + \overline{\b F_1}+\gamma$ for some constant $\gamma \in {\mathbb C}$,
the last equality yields $\a \,(F_1(z)-F_1(w))=-\,\overline{\b (F_1(z)-F_1(w))}$. Taking the modulus on both
sides, we get $F_1(z)-F_1(w)=0$, since $|\a| \neq |\b|$. The univalence of $F_1$ forces $z=w$, thus proving the claim.
\end{pf}

\subsection{Linearly independent case}\label{ssec-constantdilat}
We now investigate the generic setting.

\begin{thm}\label{thm-equalJacobian-non-constantdilat}
Let $F_1+\overline{G_1}$ and $F_2+\overline{G_2}$ be two harmonic mappings on a simply connected domain $\Omega$, with Jacobians $J_1=|F_1'|^2-|G_1'|^2$ and $J_2=|F_2'|^2-|G_2'|^2$, respectively. Assume that $F_1'$ and $G_1'$  are linearly independent. If $J_1=J_2>0$, then there exist two complex constants $\a$, $\b$ with $|\a|^2=1+|\b|^2$ and a real number $\xi$ such that
\begin{equation}\label{eq-thm-nonconstant}
\left(\begin{array}{c}
F_2'\\
G_2'
\end{array}\right)=
\left(\begin{array}{cc}
\a &  \b\\
\overline \b & \overline \a
\end{array}\right)
\left(\begin{array}{cc}
1 & 0\\
0 & e^{i\xi}
\end{array}\right)
\left(\begin{array}{c}
F_1'\\
G_1'
\end{array}\right)\,.
\end{equation}
\end{thm}

\begin{pf}
Since, by assumption, $|F_1'|>0$, we may denote $\varphi=F_2'/F_1'$. Also, the dilatation $\omega_1=G_1'/F_1'$ of $F_1+\overline{G_1}$ is not constant. Since $J_1=J_2$, we have
$$1-|\omega_1|^2=|\varphi|^2-|\psi|^2\,,$$
where $\psi=\omega_2\varphi=G_2'/F_1'$. We rewrite this as
\begin{equation}\label{eq-thm-non-constantdilat1}
1+|\psi|^2=|\omega_1|^2+|\varphi|^2\, .
\end{equation}
Note that if $\varphi$ equals a constant $k$ (necessarily different from zero), then the previous equation can be written as
\[
1+|\psi|^2=|\omega_1|^2+|k|^2\,.
\]
\par
If $|k|=1$, then $|F_2'|=|F_1'|$ and the above relation implies $|G_1'|=|G_2'|$. In view of the open mapping theorem we must have $F_2'=e^{is_1}F_1'$ and $G_2'=e^{is_2}G_1'$ for some real constants $s_1$ and $s_2$, and then (\ref{eq-thm-nonconstant}) holds with $\b=0,\, \a=e^{is_1}$, and $\xi= s_1+s_2$.
\par
If $|k|\neq 1$, then by Lemma~\ref{lem-log}  we have that $\omega_1$ is constant, which contradicts our assumption. Thus, we can assume that $\varphi$ is not constant, so that its derivative is not identically zero. Therefore, there exists a disk $D=D(z_0,R)$, centered at some $z_0\in\Omega$ and with radius $R>0$, contained in $\Omega$ such that $\varphi'(z)\neq 0$ for all $z\in D$. The forthcoming analysis is done in this disk.
\par
Taking the Laplacian of both sides of (\ref{eq-thm-non-constantdilat1}), we obtain
\[
|\psi'|^2=|\varphi'|^2+|\omega_1'|^2\,.
\]
Since $\varphi'\neq 0$, we get
\begin{equation}\label{eq-thm-non-constantdilat2}
1+\left|\frac{\omega_1'}{\varphi'}\right|^2=\left|\frac{\psi'}{\varphi'}\right|^2\,.
\end{equation}
By Lemma~\ref{lem-log} we see that both $\psi'/\varphi'$ and $\omega_1'/\varphi'$ are constant functions in $D$ and hence (by the identity principle) throughout $\Omega$. Let $m=\omega_1'/\varphi'$. Therefore (using also (\ref{eq-thm-non-constantdilat2})) we have
\[
\omega_1'=m \,\varphi'\quad {\rm and} \quad \psi'=e^{i\theta} \sqrt{1+|m|^2}\,\varphi^\prime\,,
\]
where $\theta$ is a real number. Thus, for certain complex constants $n$ and $p$,
\begin{equation}\label{eqn}
\omega_1=m \varphi+n \quad {\rm and}\quad \psi=e^{i\theta} \sqrt{1+|m|^2}\varphi+p\,.
\end{equation}
Using this information in (\ref{eq-thm-non-constantdilat1}), we get
\begin{eqnarray}
\nonumber 1+(1+|m|^2)|\varphi|^2 &+&|p|^2 + 2Re\left\{e^{i\theta}\overline p \sqrt{1+|m|^2}\varphi\right\}\\ \nonumber &=& |\varphi|^2+|m|^2|\varphi|^2+|n|^2+2Re\left\{m\overline n \varphi\right\}\,.
\end{eqnarray}
Hence
\[
2Re\left\{\left(e^{i\theta}\overline p \sqrt{1+|m|^2}-m\overline n\right)\varphi\right\}=|n|^2-|p|^2-1\,.
\]
As a consequence, we see that unless $e^{i\theta}\overline p \sqrt{1+|m|^2}-m\overline n=0$, the values of $\varphi$ lie on a line. This is not possible for non-constant $\varphi$. Therefore, we have
\begin{equation}\label{eq-realpartbetagamma}
e^{i\theta}\overline p \sqrt{1+|m|^2}=m\overline n
\end{equation}
and also
\begin{equation}\label{eq-modulusbetagamma}
|n|^2=1+|p|^2.
\end{equation}
Taking absolute values in (\ref{eq-realpartbetagamma}) gives $|p|^2(1+|m|^2)=|m|^2 |n|^2$, while from (\ref{eq-modulusbetagamma}) we have $|m|^2(1+|p|^2)=|m|^2 |n|^2$. Therefore
\[
p=e^{is}m
\]
for some real number $s$. Note that if $m=0$, then $\omega_1$ is constant, which contradicts our assumptions. Thus, $m\neq 0$ and from (\ref{eq-realpartbetagamma}) we get
\[
n=\frac{e^{-i\theta} e^{is} m\sqrt{1+|m|^2}}{\overline m}.
\]
Expressing $\omega_1$ and $\varphi$ in the first relation in (\ref{eqn}), we obtain
\[
\frac{G_1'}{F_1'}=m\,\frac{F_2'}{F_1'}+n\,.
\]
Hence
\[
G_1'=m\,F_2' + n\,F_1'
\]
and
\begin{equation}\label{eqn2}
F_2'=-\frac{n}{m}\,F_1' + \frac{1}{m}\,G_1' =- \frac{e^{-i\theta} e^{is} \sqrt{1+|m|^2}}{\overline m}\,F_1'
+ \frac{1}{m}\,G_1'\,.
\end{equation}
Making $\psi$, $\varphi$ and $p$ explicit in the second relation in (\ref{eqn}) gives
\[
\psi=\frac{G_2'}{F_1'}=e^{i\theta} \sqrt{1+|m|^2}\,\frac{F_2'}{F_1'}+e^{is}m\,.
\]
Multiplying this relation by $F_1'$ and using (\ref{eqn2}), we deduce
\begin{eqnarray*}
G_2' &=& e^{i\theta} \sqrt{1+|m|^2}\,F_2'+e^{is}m\, F_1'  \\
&=& e^{i\theta} \sqrt{1+|m|^2}\Big( - \displaystyle\frac{e^{-i\theta} e^{is} \sqrt{1+|m|^2}}{\overline m}\,F_1'+ \frac{1}{m}\,G_1'\Big)
+e^{is}m\, F_1'  \\
&=& -\displaystyle\frac{e^{is}}{\overline m}\,F_1' + \frac{e^{i\theta}  \sqrt{1+|m|^2}}{m}\,G_1'\,.
\end{eqnarray*}
Finally, setting $\xi=s+\pi$ and denoting
\[
\a:= \displaystyle\frac{ e^{i\xi} e^{-i\theta} \sqrt{1+|m|^2}}{\overline m}\,,
\quad \b:=\frac{e^{-i\xi}}{m}\,,
\]
leads to
\begin{eqnarray*}
F_2' &=&  \a\,F_1' + e^{i\xi}\, \b \,G_1'\,, \\
G_2' &=& \overline \b\,F_1' +  e^{i\xi}\,\overline \a \,G_1'\,,
\end{eqnarray*}
which is a re-expression of (\ref{eq-thm-nonconstant}).
\end{pf}
\par\smallskip
\begin{rk}
Assume that $F_1+\overline{G_1}$ and $F_2+\overline{G_2}$ are two sense preserving harmonic mappings in the simply connected domain $\Omega$, related by (\ref{eq-thm-nonconstant}). In the case when $F_1+\overline{G_1}$ is univalent and $\xi=0$, then $F_2+\overline{G_2}$ is obtained by composing with a sense preserving affine transformation. Since this affine transformation preserves univalence, $F_2+\overline{G_2}$ is univalent as well. This is not true if $\xi\neq 0$. For example, the so-called harmonic Koebe function introduced in \cite{CSS} (see also \cite[Sec. 5.3]{Dur-Harm}), defined by $K=f+\overline g$, where
\[
f(z)=\frac{z-\frac 12 z^2+\frac 16 z^3}{(1-z)^3}\quad{\rm and}\quad g(z)=\frac{\frac 12 z^2+\frac 16 z^3}{(1-z)^3}\,,\quad z\in\mathbb{D}\,,
\]
is univalent in the unit disk $\D$ but there exists $|\mu|=1$ such that $f+\overline{\mu g}$ is not univalent (see \cite[Thm.~7.1]{HM-Stable}).
\end{rk}

\section{The solutions}
In what follows we will set
\begin{equation}\label{eq-fF}
F'(t,z)=f(t,z)\quad {\rm and}\quad G'(t,z)=g(t,z)
\end{equation}
and use the notation $F_0+\overline{G_0}$ for the function $F(0, \cdot)+\overline{G(0, \cdot)}$, which is supposed to be univalent in the simply connected domain $\Omega_0$. Also, we write $f_0=f(0, \cdot)$ and $g_0= g(0, \cdot)$.

\subsection{Solutions in the linearly dependent case}
We start by finding the solutions $f\not\equiv 0$ and $g\not\equiv 0$ such that the governing equation (\ref{ge}) holds under the additional assumption that the initial harmonic (sense preserving) labelling map $F_0+\overline{G_0}$ is such that $F_0'$ and $G_0'$ are linearly dependent, or, in
other words, $F_0+\overline{G_0}$ has constant dilatation. Indeed, since the harmonic map $F_0+\overline{G_0}$ is sense preserving, the linear dependence translates into the fact that there exists a constant $\lambda\in\C$ with $|\lambda|<1$
such that $G_0'=\lambda F_0'$.
\begin{thm}\label{thm-1stcathegory}
Let $\Omega_0\subset\C$ be a simply connected domain.
If the initial harmonic (univalent, sense preserving) labelling map $F_0+\overline{G_0}$ satisfies $G_0'=\lambda F_0'$ for some $\lambda\in \C$, then
the particle motion (\ref{an}) of a fluid flow, defined by means of (\ref{eq-fF}), is given by
\begin{equation}\label{eq-formulas-1stcathegory0}
\left\{\begin{array}{l}
f(t,z)=\sqrt{1-|\lambda|^2+|\b(t)|^2}\, e^{i\int_0^t \frac{\nu_0+Im\{\b(s)\overline{\b_t(s)}\}}{1-|\lambda|^2+|\b(s)|^2}\, ds}\,F_0'(z)\,,\\
g(t,z)=\b(t)F_0'(z),
\end{array}\right.
\end{equation}
where $\b:[0,\infty)\to \C$ is a $C^1$ function with $\b(0)=\lambda$, and $\nu_0 \in \R$ is an arbitrary constant.
\end{thm}

\begin{proof} Recall that we are using the notation $f_0=F_0'$. Since the Jacobian of the labelling map remains unchanged at all times $t$, we can apply Theorem~\ref{lin-dep} to deduce that
$f(t,z)=\a(t)f_0(z)$ and $g(t,z)=\b(t)\, f_0(z)$, where $|\a|^2-|\b|^2\equiv 1-|\lambda|^2>0$.  Using this in (\ref{ge}),  we obtain
\[
f_t \,\overline{f} - g\, \overline{g}_t = (\a_t\,\overline{\a} -\b\,  \overline{\b}_t) \,|f_0(z)|^2 = i\,\nu(z,\,\overline{z})\,.
\]
As $F_0+\overline{G_0}$ is a sense preserving mapping, we have $|f_0|>0$. Thus $\a_t\,\overline{\a} -\b\, \overline{\b}_t=i\,\nu_0$ for some constant $\nu_0 \in \R$,
and we are lead to the system
\begin{equation}\label{2e}
\left\{\begin{array}{l}
\a_t\,\overline{\a} -\b  \overline{\b}_t =i\,\nu_0\,,\\
|\a|^2=|\b|^2+c\,,\
\end{array}\right.
\end{equation}
where $c=1-|\lambda |^2>0$.
The second equation above ensures that $|\a|>0$, which allows us to write $\a(t)=R(t)\,e^{i\varphi(t)}$ for appropriate $C^1$-functions $R: [0,\infty) \to (0,\infty)$ and $\varphi: [0,\infty) \to \R$ (see, for instance, \cite[Thm. 2.24]{Kuhnel}).  The system (\ref{2e}) written in polar coordinates
becomes
\begin{equation}\label{22e}
\left\{\begin{array}{l}
R_tR + iR^2 \varphi_t - \b\, \overline{\b}_t=i\,\nu_0\,,\\
R=\sqrt{|\b|^2+c}\,.\\
\end{array}\right.
\end{equation}

The first equation in (\ref{22e}) in conjunction with the time-differentiation of the second equation in (\ref{22e}) yields $R^2\,\varphi_t=\nu_0 + Im\{\b\,\overline{\b}_t\}$. Thus (\ref{2e}) reduces to
$$\left\{\begin{array}{l}
\varphi(t)=\varphi(0) + \int_0^t \frac{\nu_0 + Im\{\b(s)\,\overline{\b}_t(s)\}}{|\b(s)|^2+c}\,ds\,,\\
R=\sqrt{|\b|^2+c}\,.\\
\end{array}\right.$$
The initial conditions $f(0,\cdot)=f_0, \, g(0,\cdot)=g_0$ now give $\b(0)=\lambda$ and $\varphi(0)=0$. Therefore
we obtain (\ref{eq-formulas-1stcathegory0}).
\end{proof}
\par
\begin{ex}  Kirchhoff's solution \cite{Kirchhoff} is the particular case of (\ref{eq-formulas-1stcathegory0}) in which
\[
F_0'= A e^{ikz}\,,\quad \b(t) \equiv \lambda \,,\quad {\rm and}\quad \nu_0=0\,,
\]
where $A$ and $k$ are non-zero real constants and $\lambda\in (0, 1)$. The condition on the univalence of $F_0$ requires that $\Omega_0$ does not contain points $z$ and $w$ with
\[
Im\{z\}=Im\{w\}\quad{\rm and }\quad Re\{z\}=Re\{w\}+\frac{2m\pi}{k}
\]
for some integer $m$.
\end{ex}

\subsection{The solutions in the linearly independent case}\label{ssec-secondclasssols}
Using again the notation $F_0+\overline{G_0}$ for the function $F(0, \cdot)+\overline{G(0, \cdot)}$ and keeping in mind that $F_0+\overline{G_0}$ is sense preserving univalent in the simply connected domain $\Omega_0$, we have that $|F_0'|^2-|G_0'|>0$ in $\Omega_0$. Also, as before, we set $f_0=f(0, \cdot)$ and $g_0= g(0, \cdot)$ (hence $|f_0|^2-|g_0|^2\neq 0$).
\par
Now, we will consider solutions $f\not\equiv 0$ and $g\not\equiv 0$ such that (\ref{ge}) holds and such that $F_0+\overline{G_0}$ has non-constant dilatation.
\par
\begin{thm}\label{thm-2ndcathegory}
Let $\Omega_0\subset\C$ be a simply connected domain. Assume that the initial harmonic (univalent, sense preserving) labelling map $F_0+\overline{G_0}$ is such that $F_0'$ and  $G_0'$ are linearly independent. The particle motion (\ref{an}) of a fluid flow, defined by means of (\ref{eq-fF}), is either described by
\begin{equation}\label{eq-formulas-2ndcathegory}
\left\{\begin{array}{l}
f(t,z) = \sqrt{1+|\b(t)|^2}\, e^{i\int_0^t \frac{\nu_0+Im\{\b_t(s)\overline{\b(s)}\}}{1+|\b(s)|^2}\, ds }\,F_0'(z) +  \b(t)\,G_0'(z)\,, \\
g(t,z)=  {\overline{\b(t)}}\,F_0'(z) + \sqrt{1+|\b(t)|^2}\, e^{-i\int_0^t \frac{\nu_0+Im\{\b_t(s)\overline{\b(s)}\}}{1+|\b(s)|^2}\, ds }\,G_0'(z)\,,
\end{array}\right.
\end{equation}
where $\b:[0,\infty)\to \C$ is a $C^1$ function and $\nu_0 \in \R$, or by
\begin{equation}\label{eq-formulas-3rdcathegory}
\left\{\begin{array}{l}
f(t,z) = e^{i\nu_0 t}\,F_0'(z)\,,\\
g(t,z)=  e^{i(\xi_0-\nu_0)t}\,G_0'(z)\,,
\end{array}\right.
\end{equation}
where $\nu_0$ is as above and $\xi_0\in\R\neq \{0\}$.
\par
Moreover, for the solutions (\ref{eq-formulas-2ndcathegory}), univalence at any time is ensured once it holds at time $t=0$. Univalence holds for the solutions (\ref{eq-formulas-3rdcathegory}) if and only if  $F_0+\overline{\lambda G_0}$ is univalent for all $\lambda$ with $|\lambda|=1$.
\end{thm}
\par
\begin{pf}
By Theorem~\ref{thm-equalJacobian-non-constantdilat}, we know that there exist $C^1$ functions $\a, \b :[0, \infty)\to \C$ with $|\a(t)|^2-|\b(t)|^2=1$ for all $t\geq 0$ and $\xi:[0,\infty)\to\R$ such that
\begin{equation}\label{eq-system}
\left\{\begin{array}{l}
f(t,z) = \a(t)\,f_0(z) + e^{i\xi(t)}\,  \b(t)\,g_0(z)\,, \\
g(t,z)=  {\overline{\b(t)}}\,f_0(z) +  e^{i\xi(t)}\, {\overline{\a(t)}}\,g_0(z)\,.
\end{array}\right.
\end{equation}
Note that since $f(0, \cdot)=f_0$, $g(0, \cdot)=g_0$, and the dilatation of $F_0+\overline{G_0}$ is non-constant, we have the initial conditions
\begin{equation}\label{eq-normalizations}
\a(0)=1\,,\quad \b(0)=0\,,\quad{\rm and}\quad \xi(0)=0\,.
\end{equation}
\par
A straightforward calculation shows that
\begin{eqnarray}\label{eq-function0}
 f_t \overline f &-&g\, \overline g_t= (\a_t\, \overline \a- \b_t\, \overline \b) (|f_0|^2-|g_0|^2) + 2iRe\{\xi_t \a\overline \b e^{-i\xi}f_0\overline{g_0}\}
\\ \nonumber &+& i\xi_t\, (|\a|^2+|\b|^2)\, |g_0|^2
\\
\nonumber &=& (\a_t\, \overline \a - \b_t\, \overline \b -i\xi_t|\a|^2) (|f_0|^2-|g_0|^2) + i \xi_t|\a f_0+\b e^{i\xi}g_0|^2\,.
\end{eqnarray}
Now, by (\ref{ge}), we know that $f_t\, \overline f  -g\, \overline g_t=i\nu(z,\overline z)$. Since the function $|f_0|^2-|g_0|^2$ only depends on $z$ and $\overline z$ and is always different from zero, we obtain that
\begin{eqnarray}\label{eq-beforedilat}
\nonumber (\a_t(t)\overline{\a(t)}&-&\b_t(t)\overline{\b(t)}-i\xi_t(t)|\a(t)|^2)  \\ &+& i \frac{\xi_t(t)|\a(t)f_0(z)+ \b(t)e^{i\xi(t)}g_0(z)|^2}{|f_0(z)|^2-|g_0(z)|^2}=i\widetilde\nu(z,\overline z)\,,
\end{eqnarray}
where $\widetilde \nu=\nu/((|f_0|^2-|g_0|)^2)$. Let us re-write the previous equation using that $|f_0|^2-|g_0|^2=|f_0|^2(1-|\omega|^2)$, where $\omega=g_0/f_0$ is the dilatation of $F_0+\overline{G_0}$. Note that since we are assuming that $\omega$ is not constant, we have that there exists an open disk $D\subset \Omega_0$ such that $\omega'(z)\neq 0$ for all $z\in \Omega_0$. From now on, we will assume that $z\in D$. Within these terms (\ref{eq-beforedilat}) becomes
\begin{eqnarray}\label{eq-afterdilat1}
\nonumber (\a_t(t)\overline{\a(t)}&-&\b_t(t)\overline{\b(t)}-i\xi_t(t)|\a(t)|^2)   \\ &+& i \frac{\xi_t(t)|\a(t)+\b(t)e^{i\xi(t)}\omega(z)|^2}{1-|\omega(z)|^2}=i\widetilde\nu(z,\overline z)\,.
\end{eqnarray}
Taking derivatives with respect to $z$ in (\ref{eq-afterdilat1}), we obtain that the function
\begin{eqnarray*}\label{eq-function1}
&i& \xi_t(t) \frac{\omega'(z)}{(1-|\omega(z)|^2)^2}\, \left(\b(t)\,e^{i\xi(t)} +\a(t)\overline{\omega(z)}\right)
\, \left(\overline{\a(t)}+\overline{\b(t)}e^{-i\xi(t)}\overline{\omega(z)}\right)\\
&=& i\xi_t(t) \left(\overline{\a(t)}\b(t)e^{i\xi(t)}+(|\a(t)|^2+|\b(t)|^2)\overline{\omega(z)}+\a(t)\overline{\b(t)}e^{-i\xi(t)}\overline{\omega(z)}^{\, 2}\right)\\
&\times& \frac{\omega'(z)}{(1-|\omega(z)|^2)^2}
\end{eqnarray*}
only depends on $z$ and on $\overline z$. Hence, so does
\[
i\xi_t(t) \left(\overline{\a(t)}\b(t)e^{i\xi(t)}+(|\a(t)|^2+|\b(t)|^2)\overline{\omega(z)}+\a(t)\overline{\b(t)}e^{-i\xi(t)}\overline{\omega(z)}^{\, 2}\right)\,,
\]
which has derivative with respect to $\overline z$ equal to
\[
i\xi_t(t)(|\a(t)|^2+|\b(t)|^2)\overline{\omega'(z)}+2 i\xi_t(t) \a(t)\overline{\b(t)}e^{-i\xi(t)}\overline{\omega(z)} \overline{\omega'(z)}\,.
\]
Since $\omega'$ is supposed to be non-zero, we have that
\begin{equation}\label{eq-function4}
i\xi_t(t)(|\a(t)|^2+|\b(t)|^2)+2 i\xi_t(t) \a(t)\overline{\b(t)}e^{-i\xi(t)}\overline{\omega(z)}=\gamma(z,\overline z)\,,
\end{equation}
where $\gamma$ is a function of $z$ and $\overline z$. This fact finally implies (again taking derivatives with respect to $\overline z$ in (\ref{eq-function4}) and dividing out by $\omega'$) the equation
\[
\xi_t \a\overline{\b}e^{-i\xi}=c_1
\]
for some constant $c_1$ and, in particular, by using this information in (\ref{eq-function4}) we get that $(|\a|^2+|\b|^2)\xi_t$ is constant as well. Moreover, by (\ref{eq-function0}), we see that $\a\overline \a_t-\b\overline \b_t$ is constant too (indeed, using that $|\a|^2-|\b|^2=1$, it is easy to check that $\a_t\,\overline \a-\b_t\,\overline \b$ is a purely imaginary complex number). In other words, we have seen that the following system must be satisfied by the functions $\a$, $\b$, and $\xi$ (here, $c_1$, $c_2$, and $c_3$ are certain constants with $c_3\in\R$):
\begin{equation}\label{eq-systemabxi}
\left\{\begin{array}{l}
\xi_t \a\overline{\b}e^{-i\xi}=c_1\,,\\
(|\a|^2+|\b|^2)\xi_t=c_2\,,\\
\a_t \overline \a-\b_t\overline \b=ic_3\,,\\
|\a|^2-|\b|^2=1\,.
\end{array}\right.
\end{equation}
Now, note that the second and fourth equations in (\ref{eq-systemabxi}) can be re-written as
\[
\left\{\begin{array}{l}
|\a|^2\xi_t+|\b|^2\xi_t=c_2\,,\\
|\a|^2\xi_t-|\b|^2\xi_t=\xi_t\,,
\end{array}\right.
\]
which gives
\begin{equation}\label{eq-systemxiconstant2}
\left\{\begin{array}{l}
2 |\a|^2\xi_t=c_2+\xi_t\,,\\
2 |\b|^2\xi_t=c_2-\xi_t\,.
\end{array}\right.
\end{equation}
On the other hand, using the first equation in (\ref{eq-systemabxi}), a direct consequence of (\ref{eq-systemxiconstant2}) is that
\begin{eqnarray*}
4|c_1|^2=4(\xi_t \a\overline{\b}e^{-i\xi})\overline{(\xi_t \a\overline{\b}e^{-i\xi})}=2|\a|^2\xi_t\cdot 2 |\b|^2\xi_t= c_2^2-\xi_t^2\,.
\end{eqnarray*}
This shows that $\xi_t^2$ (hence $\xi_t$) must be a constant function.
\par\smallskip
We distinguish between two types of solutions.
\par\smallskip
\textbf{Case 1.} $\xi_t\equiv 0$. Then $\xi$ must be constant and hence, by (\ref{eq-normalizations}), we see that $\xi\equiv 0$.  Moreover, in this case (\ref{eq-systemabxi}) becomes
\[
\left\{\begin{array}{l}
\a_t \overline \a-\b_t\overline \b=ic_3\,,\\
|\a|^2-|\b|^2=1\,.
\end{array}\right.
\]
Notice that the above system is a particular case of (\ref{2e}), with $\nu_0=c_3$, $c=1$,  and where $\b$ is replaced by $\bar \b$. Taking into account that
$\a(0)=0$, the previous approach shows that its solution is given by
$$\a(t)=\sqrt{|\b(t)|^2+1}\,\exp\Big(i\int_0^t \frac{c_3+ Im\{\b_t(s) \,\overline{\b(s)}\}}{1+|\b(s)|^2}\, ds\Big)\,,$$
and $\b:[0,\infty) \to \C$ is an arbitrary $C^1$-function with $\b(0)=0$.
These considerations prove the first part of the theorem. The univalence is obtained as in Remark 2.
\par\smallskip
\textbf{Case 2.} $\xi_t\equiv \xi_0,\ \xi_0\neq 0$. Note that in this case (using that, by (\ref{eq-normalizations}), $\xi(0)=0$),  we have $\xi(t)=\xi_0 t$.  By (\ref{eq-systemxiconstant2}) we see that both functions $|\a|$ and $|\b|$ are constants. A further application of (\ref{eq-normalizations}) gives that $\b\equiv 0$ and $\a(t)=e^{i\varphi(t)}$ for some $C^1$ function $\varphi:[0, \infty)\to \R$ with $\varphi(0)=0$. Moreover, by (\ref{eq-systemabxi}) we get $\varphi'\equiv c_3$\,, so that $\varphi(t)=c_3 t$. Therefore, according to  (\ref{eq-system}), we have that
 \[
\left(\begin{array}{c}
F'(t,z)\\
G'(t,z)
\end{array}\right)=
\left(\begin{array}{cc}
1 & 0 \\
0 & e^{i\xi_0 t}
\end{array}\right)
\left(\begin{array}{cc}
e^{ic_3 t} & 0\\
0 & e^{-ic_3 t}
\end{array}\right)
\left(\begin{array}{c}
F_0'(z)\\
G_0'(z)
\end{array}\right)\,,
\]
which is  (\ref{eq-formulas-3rdcathegory}) with $\nu_0=c_3$\,.
\par\smallskip
Concerning the issue of univalence for the flows  (\ref{eq-formulas-3rdcathegory}), note that
\[
\left(\begin{array}{c}
F(t,z)\\
G(t,z)
\end{array}\right)=
\left(\begin{array}{cc}
1 & 0 \\
0 & e^{i\xi_0 t}
\end{array}\right)
\left(\begin{array}{cc}
e^{i\nu_0 t} & 0\\
0 & e^{-i\nu_0 t}
\end{array}\right)
\left(\begin{array}{c}
F_0(z)\\
G_0(z)
\end{array}\right) + \left(\begin{array}{c}
\mu(t)\\
\nu(t)
\end{array}\right)\,,
\]
where $\mu$ and $\nu$ are real-valued $C^1$ functions in $[0, \infty)$. That is,
\[F(t,z)+\overline{G(t,z)}=e^{i\nu_0 t}\left(F_0(z)+e^{-i\xi_0 t}\overline{G_0(z)}\right)+v(t)\,,
\]
where $v(t)$ is the translation vector $(\mu(t), \overline{\nu(t)})$. The translation vector $v$, as well as the multiplication factor $e^{i\nu_0 t}$, play no role regarding univalence. Since $e^{-i\xi_0 t}$ covers the unit circle as $t\in [0,\infty)$, the claim is proved.
\end{pf}
\par
\begin{ex}  Gerstner's flow \cite{Gerstner} corresponds to the case of  (\ref{eq-formulas-3rdcathegory}) in which
\[
\nu_0=0\,,\quad \xi_0=\sqrt{k\, \mathfrak{g}}\,, \quad   F_0' = 1\,,\quad {\rm and}\quad  G_0'(z)=-e^{-ikz}\,,
\]
where $k>0$, $\mathfrak{g}$ is the gravitational constant of acceleration, and $z\in\Omega_0=\{z\in\C\colon Im\{z\}<0\}$.
\end{ex}
\par\smallskip

\section*{Acknowledgements}
This research was developed during the second author's research stay at the University
of Vienna funded by the ERC Advanced Grant ``Nonlinear studies of water
flows with vorticity''. She wants to use this opportunity to thank the faculty and staff members at the Faculty of Mathematics for their hospitality.


\end{document}